\documentclass[letterpaper,10pt,conference]{ieeeconf}

\usepackage{cite}
\usepackage{graphicx}
\usepackage{tikz}
\usepackage{subcaption}
\usepackage[cmex10]{amsmath}
\usepackage{algpseudocode}
\usepackage{algorithmicx} 
\usepackage{array}
\usepackage{url}

\usepackage{epsfig}
\usepackage{amsfonts,amssymb,multirow,bigstrut,booktabs,ctable,latexsym}

\usepackage[mathscr]{eucal}
\usepackage{verbatim}

\usepackage{amsthm}
\usepackage{algorithm}

\usepackage{stmaryrd}

\IEEEoverridecommandlockouts  
\begin{document}

\newtheorem{definition}{Definition}
\newtheorem{lemma}{Lemma}
\newtheorem{corollary}{Corollary}
\newtheorem{theorem}{Theorem}
\newtheorem{example}{Example}
\newtheorem{proposition}{Proposition}
\newtheorem{remark}{Remark}
\newtheorem{assumption}{Assumption}
\newtheorem{corrolary}{Corrolary}
\newtheorem{property}{Property}
\newtheorem{ex}{EX}
\newtheorem{problem}{Problem}
\newcommand{\argmin}{\arg\!\min}
\newcommand{\argmax}{\arg\!\max}
\newcommand{\st}{\text{s.t.}}
\newcommand \dd[1]  { \,\textrm d{#1}  }

\title{\Large\bf Control Barrier Functions for Abstraction-Free Control Synthesis under Temporal Logic Constraints}

\author{Luyao Niu and Andrew Clark %
\thanks{L. Niu and A. Clark are with the Department of Electrical and Computer Engineering, Worcester Polytechnic Institute, Worcester, MA 01609 USA.
{\tt\small \{lniu,aclark\}@wpi.edu}}
	\thanks{This work was supported by the National Science Foundation and the Office of Naval Research via grants CNS-1941670 and N00014-17-S-B001.}
}
\thispagestyle{empty}
\pagestyle{empty}

\maketitle

\begin{abstract}
Temporal logic has been widely used to express complex task specifications for cyber-physical systems (CPSs). One way to synthesize a controller for CPS under temporal logic constraints is to first abstract the CPS as a discrete transition system, and then apply formal methods. This approach, however, is computationally demanding and its scalability suffers due to the curse of dimensionality. In this paper, we propose a control barrier function (CBF) approach to abstraction-free control synthesis under a linear temporal logic (LTL) constraint. We first construct the deterministic Rabin automaton of the specification and compute an accepting run. We then compute a sequence of LTL formulae, each of which must be satisfied during a particular time interval, and prove that satisfying the sequence of formulae is sufficient to satisfy the LTL specification. Finally, we compute a control policy for satisfying each formula by constructing an appropriate CBF. We present a quadratic program to compute the controllers, and show the controllers synthesized using the proposed approach guarantees the system to satisfy the LTL specification, provided the quadratic program is feasible at each time step. A numerical case study is presented to demonstrate the proposed approach.
\end{abstract}

\section{Introduction}\label{sec:intro}
Cyber-physical systems (CPSs) are assigned increasingly complex objectives including reactive and sequential tasks in domains such as autonomous vehicles, advanced manufacturing, and health care systems. Temporal logics \cite{baier2008principles} such as linear temporal logic (LTL) are widely adopted to specify properties and verify behaviors of CPSs due to their rich and rigorous expressiveness. Consequently, temporal logic based control has gained research attention in areas including robotics \cite{kress2009temporal} and traffic network control \cite{coogan2015traffic}.

While control synthesis for CPSs focuses on generating closed-loop controllers in the continuous domain, the semantics of LTL are defined in a discrete manner, which are normally captured by finite state automata. Motivated by formal methods in the context of model checking \cite{baier2008principles}, control synthesis for CPSs are normally lifted from continuous domain to discrete domain by computing a finite abstraction of the CPS. The finite abstraction can be generated by partitioning the state space and control input space of CPSs with consistency guarantees (e.g., simulation and bisimulation relations \cite{alur2000discrete}). Typical forms of finite abstractions of CPSs include finite transition system \cite{fainekos2009temporal}, Markov decision process\cite{ding2014optimal,wongpiromsarn2009receding}, and stochastic game \cite{niu2019optimal}.

Abstraction-based approaches are computationally demanding and suffer from the curse of dimensionality. To avoid the computation required for constructing the finite abstraction, researchers have investigated control synthesis on the continuous state space. Techniques include formulating the LTL constraint as a mixed integer linear program \cite{wolff2014optimization}, a sequence of stochastic reachability problems \cite{horowitz2014compositional}, and a mixed continuous-discrete HJB equation \cite{papusha2016automata}. Solving for the controller thus requires numerical methods which can be computationally expensive.

In this paper, we consider control synthesis for CPSs under LTL constraints without computing a finite abstraction. We aim to compute a feedback controller such that a control affine system satisfies a given LTL specification from the fragment of LTL without next operator. We present a control barrier function (CBF) based framework to compute the controller. There are two main advantages of our approach compared to the state of the art. First, we avoid both finite-state abstraction of the CPS and approximate solution of the HJB equation, and thus reduce the computational complexity. Second, we introduce time-varying \emph{guard functions} that render our approach feasible for a broad class of LTL properties. 
To summarize, this paper makes the following contributions.
\begin{itemize}
    \item We present a CBF-based approach to synthesize a controller for CPSs under LTL constraints. The proposed approach is provably correct and avoids explicit construction of a finite abstraction of CPSs.
    \item We construct a sequence of formulae that correspond to an accepting trace of the CPS, and develop a methodology to construct CBFs for each formula. When designing the CBF, we introduce the concept of guard function, which implicitly encodes the time that each formula needs to be satisfied and enhances the feasibility of the CBF constraints.
    \item We compute the set of controllers that satisfy each constructed formula using two types of time varying CBFs. We show that by satisfying the CBF constraints for all formulae, the LTL specification is satisfied.
    \item A numerical case study on robotic motion planning is presented as evaluation. The proposed approach successfully synthesizes a controller for a specification that is infeasible under the state of art.
\end{itemize}

The remainder of this paper is organized as follows. Section \ref{sec:related} reviews the related work. Section \ref{sec:preliminary} gives preliminary background on CBFs and LTL. The problem of interest is formulated in Section \ref{sec:formulation}. Our solution approach is presented in Section \ref{sec:solution}. Section \ref{sec:simulation} gives a numerical case study on robotic motion planning. Section \ref{sec:conclusion} concludes the paper.
\section{Related Work}\label{sec:related}

Verification and control synthesis methodologies for CPS under temporal logic constraints have been proposed based on constructing an abstraction of the dynamical system \cite{kress2009temporal,fainekos2009temporal,ding2014optimal,wongpiromsarn2009receding,niu2019optimal}. However, it is normally computationally expensive to construct an abstraction of the dynamical system. 


To avoid constructing the finite abstraction, controller synthesis in continuous domain under LTL constraints has been studied. An optimization based approach was proposed in \cite{wolff2014optimization}. The authors encoded the LTL constraints as mixed-integer linear constraints. The authors of  \cite{horowitz2014compositional} studied control synthesis under co-safe LTL specifications. They computed the controller by solving a sequence of stochastic constrained reachability problems, which relied on solving nonlinear PDEs. The authors of \cite{papusha2016automata} formulated a hybrid system using the automaton and the continuous system, and computed the controller by solving a mixed continuous-discrete HJB equation. Normally PDEs and HJB equations are approximately solved using computationally expensive numerical approaches such as Finite-Difference \cite{kappen2005linear}. In this work, we give a different CBF based approach to synthesize a controller by solving a quadratic program at each time.

The concept of CBF extends barrier function to CPSs with control inputs \cite{wieland2007constructive}. Recent works have developed CBF based approaches for safety critical system \cite{wang2017safety,ames2016control}. LTL specifications, however, capture a richer set of properties that cannot be modeled by safety constraints alone.

CBFs have gained popularity for CPSs under LTL constraints \cite{srinivasan2019control,srinivasan2018control} and signal temporal logic (STL) constraints \cite{lindemann2018control,yang2019continuous}. The authors of \cite{srinivasan2019control,srinivasan2018control} studied the problem of control synthesis of multi-robot system under LTL constraints using CBFs. Under their approaches, they solved a sequence of reachability problems using CBFs, and proposed a relaxation of CBF based constraints.
This work differs from \cite{srinivasan2019control,srinivasan2018control} in the following two aspects. First, our approach considers a broader fragment of LTL compared to \cite{srinivasan2019control}. A fragment of LTL specification named $\text{LTL}_{Robotic}$ is considered in \cite{srinivasan2019control}, while we consider LTL without next operator. Second, we propose a different approach to resolve the infeasibility between CBF constraints compared to \cite{srinivasan2019control,srinivasan2018control}. As we will demonstrate later, our proposed approach synthesizes a controller for a specification that is infeasible using the approaches in \cite{srinivasan2019control,srinivasan2018control}, and thus serves as a complement to \cite{srinivasan2019control,srinivasan2018control}.

\section{Preliminaries}\label{sec:preliminary}

A function $f:\mathbb{R}^n\times [0,\infty)\mapsto\mathbb{R}$ is Lipschitz continuous on $\mathcal{X}\subset\mathbb{R}^n$ if there exists some constant $K>0$ such that $\|f(x_1)-f(x_2)\|\leq K\|x_1-x_2\|$ for all $x_1,x_2\in\mathcal{X}$, where $\|\cdot\|$ is the Euclidean norm. A function $f$ is locally Lipschitz continuous if there exist constants $\tau>0$ and $K>0$ such that $\|f(x_1)-f(x_2)\|\leq K\|x_1-x_2\|$ for all $\|x_1-x_2\|\leq\tau$.  

A continuous function $\alpha:[0,a)\mapsto[0,\infty)$ belongs to class $\mathcal{K}$ if it is strictly increasing and $\alpha(0)=0$. A continuous function $\alpha:[-b,a)\mapsto(-\infty,\infty)$ is said to belong to extended class $\mathcal{K}$ if it is strictly increasing and $\alpha(0)=0$ for some $a,b>0$.
The following lemma gives smooth approximations of $\min$ and $\max$ operators.
\begin{lemma}[Approximation of $\min$ and $\max$ Operators \cite{boyd2004convex}]\label{lemma:min operator}
Consider a set of functions $h_i(x,t)$. Then for $\lambda>0$
\begin{align*}
    \min_i h_i(x,t)\geq -\ln{\left(\sum_{i=1}^k\exp{(-h_{i}(x,t))}\right)},\\
    \max_i h_i(x,t)\geq \frac{\sum_ih_i(x,t)\exp(\lambda h_i(x,t))}{\sum_i\exp(\lambda h_i(x,t))}.
\end{align*}
\end{lemma}

\subsection{Linear Temporal Logic without Next Operator}

In this subsection, we present background on LTL without next operator, denoted as $\text{LTL}_{\setminus\bigcirc}$. An $\text{LTL}_{\setminus\bigcirc}$ formula is defined inductively as  $\varphi=\top\mid\pi\mid\neg\varphi\mid\varphi_1\land\varphi_2\mid\varphi_1\mathbf{U}\varphi_2$ \cite{baier2008principles}.
Other Boolean and temporal operators are defined inductively. LTL specification with implication ($\phi\implies\psi$) is defined as $\neg\phi\land\psi$. LTL specification involving eventually operator $\Diamond\pi$ is equivalent to $\top\mathbf{U}\pi$. LTL specification involving always operator $\Box\pi$ can be rewritten as $\neg\Diamond\neg\pi$. When the context is clear, we abuse the notions of LTL specification and $\text{LTL}_{\setminus\bigcirc}$ in the remainder of this paper.

Given an atomic proposition set $\Pi$, we augment $\Pi$ so that every $\pi\in\Pi$ implies $\neg\pi\in\Pi$. We interpret LTL formulae over infinite words in $2^\Pi$. Formula $\phi\land\psi$ is true iff $\phi$ and $\psi$ are both true. Formula $\phi\mathbf{U}\psi$ is true iff $\phi$ remains true until $\psi$ becomes true. A word $\eta$ that satisfies an LTL formula $\phi$ is denoted as $\eta\models\phi$.

A deterministic Rabin automaton (DRA) can be used to represent an LTL formula $\varphi$. A DRA is defined as follows.
\begin{definition}[Deterministic Rabin Automaton (DRA) \cite{baier2008principles}]\label{def:DRA}
A Deterministic Rabin Automaton (DRA) is a tuple $\mathcal{A}=(Q,\Sigma,\delta,q_0,F)$, where $Q$ is a finite set of states, $\Sigma$ is the finite set of alphabet, $\delta:Q\times\Sigma\mapsto Q$ is a finite set of transitions, $q_0\in Q$ is the initial state, and $F=\{(B(1),\Gamma(1)),\ldots,(B(S),\Gamma(S))\}$ is a finite set of Rabin pairs such that $B(s),\Gamma(s)\subseteq Q$ for all $s=1,\ldots,S$ with $S$ being a positive integer.
\end{definition}
Every LTL formula can be represented by a DRA. Translation from LTL to DRA has been extensively studied \cite[Ch.5]{baier2008principles}, and is beyond the scope of this paper. Given a state $q\in Q$, we define the set of neighbor states of $q$ as $\mathcal{N}(q)=\underset{\sigma\in\Sigma}{\cup}\delta(q,\sigma)\setminus \{q\}$, i.e., a state $q'\neq q$ is a neighbor state of $q$ if there exists some transition $\delta$ such that $q$ can transit to $q'$. A run $\eta$ on $\mathcal{A}$ is an accepting run if there exists a pair $(B(s),\Gamma(s))$ such that $\eta$ intersects with $B(s)$ finitely many times and intersects with $\Gamma(s)$ infinitely many times. Given the DRA $\mathcal{A}$, we can find the set of accepting runs on $\mathcal{A}$ by a graph search algorithm \cite{baier2008principles}. Any accepting run of the automaton can be divided into two parts including a prefix $\eta_{pref}$ and a suffix $\eta_{suff}$ that is repeated infinitely often \cite[Ch.4]{baier2008principles}. Both the prefix and suffix can be regarded as finite runs, and $\eta$ can be represented as $\eta=\eta_{pref}(\eta_{suff})^\omega$.

\section{Problem Formulation}\label{sec:formulation}

Consider a continuous-time control-affine system
\begin{equation}\label{eq:dynamic}
    \dot{x}=f(x)+g(x)u,\quad x(0)=x_0,
\end{equation}
where $x\in\mathcal{X}\subseteq\mathbb{R}^n$ is the system state, and $u\in\mathcal{U}\subseteq\mathbb{R}^m$ is input provided by the controller. The initial state at time $t=0$ is denoted as $x(0)=x_0$. Vector fields $f$ and $g$ are locally Lipschitz continuous. Given the current system state $x$, a feedback controller is a function $\mu:\mathcal{X}\times[0,\infty)\mapsto\mathcal{U}$. 


Let $\Pi$ be a finite set of atomic propositions. We define a labeling function $L:\mathcal{X}\mapsto 2^\Pi$ that maps any state $x\in\mathcal{X}$ to a subset of atomic propositions that hold true at $x$. We also define $\llbracket\pi\rrbracket=\{x|\pi\in L(x)\}$ to be the set of states that satisfies the atomic proposition $\pi\in\Pi$. In this work, we assume that $\llbracket\pi\rrbracket$ is a closed set for all $\pi\in\Pi$, and $\llbracket\pi\rrbracket$ can be represented as $\llbracket\pi\rrbracket=\{x|Z_\pi(x)\geq 0\}$, where $Z_\pi:\mathbb{R}^n\mapsto\mathbb{R}$ is a bounded and continuously differentiable function. 
We slightly overload the notation $\llbracket\cdot\rrbracket$, and define the states that satisfy a subset of atomic propositions $P\in 2^\Pi$ as 
\begin{equation}\label{eq:subset atomic}
    \llbracket P\rrbracket=\begin{cases}
    \mathcal{X}\setminus\cup_{\pi\in\Pi}\llbracket\pi\rrbracket&\mbox{ if }P=\emptyset\\
    \cap_{\pi\in P}\llbracket\pi\rrbracket\setminus\cup_{\pi\in\Pi\setminus P}\llbracket\pi\rrbracket&\mbox{ otherwise}
    \end{cases}
\end{equation}
That is, $\llbracket P\rrbracket$ is the subset of system states $\mathcal{X}$ that satisfy all and only propositions in $P$ \cite{wongpiromsarn2015automata}.

We define the trajectory of system \eqref{eq:dynamic} as $\mathbf{x}:[0,\infty)\mapsto\mathcal{X}$ that maps from any time $t\geq0$ to the system state $x(t)$. 
We then define the trace of a trajectory $\mathbf{x}$ as follows.
\begin{definition}[Trace of Trajectory \cite{wongpiromsarn2015automata}]\label{def:trace}
An infinite sequence $Trace(\mathbf{x})=P_0,P_1,\ldots$, where $P_i\in 2^\Pi$ for all $i=0,1,\ldots$ is a trace of a trajectory $\mathbf{x}$ if there exists an associated time sequence $t_0,t_1,\ldots$ of time instants such that
\begin{enumerate}
    \item $t_0=0$
    \item $t_\tau\rightarrow\infty$ as $\tau\rightarrow\infty$
    \item $t_i<t_{i+1}$
    \item $x(t_i)\in \llbracket P_i\rrbracket$
    \item if $P_i\neq P_{i+1}$, then there exists some $t_i'\in[t_i,t_{i+1}]$ such that $x(t)\in\llbracket P_i\rrbracket$ for all $t\in(t_i,t_i')$, $x(t)\in\llbracket P_{i+1}\rrbracket$ for all $t\in(t_i',t_{i+1})$, and either $x(t_i')\in\llbracket P_i\rrbracket$ or $x(t_i')\in\llbracket P_{i+1}\rrbracket$.
\end{enumerate}
\end{definition}
The trace of the system trajectory gives the sequence of atomic propositions satisfied by the system, and thus bridges the system behavior with temporal logic specification. Given a controller $\mu$, we denote the trajectory under controller $\mu$ as $\mathbf{x}^\mu$. The trace of trajectory $\mathbf{x}^\mu$ is denoted by $Trace(\mathbf{x}^\mu)$. Suppose a specification $\varphi$ belonging to $\text{LTL}_{\setminus\bigcirc}$ is given to system \eqref{eq:dynamic}. If $Trace(\mathbf{x}^\mu)\models\varphi$, we say system \eqref{eq:dynamic} satisfies $\varphi$ under controller $\mu$, or controller $\mu$ satisfies $\varphi$. We state the problem of interest as follows.
\begin{problem}\label{prob:formulation}
Compute a feedback controller $\mu:\mathcal{X}\times[0,\infty)\mapsto\mathcal{U}$ under which system \eqref{eq:dynamic} satisfies the given LTL specification $\varphi$ belonging to $LTL_{\setminus\bigcirc}$. That is, compute a controller $\mu$ such that $Trace\left(\mathbf{x}^\mu\right)\models\varphi$.
\end{problem}

\section{Solution Approach}\label{sec:solution}

We present a framework to solve Problem \ref{prob:formulation} in this section. We first introduce two types of CBFs. Given an LTL specification $\varphi$, we then present how to design CBFs using the automaton of the LTL specification. We construct a sequence of LTL formulae that correspond to an accepting run on the DRA of the LTL specification. Then we define a time varying CBF for each formula. We show that satisfying each formula is equivalent to guaranteeing the positivity of the corresponding CBF. Then we compute the controllers that ensure the CBF associated with each formula to be positive, and hence satisfies the LTL specification.

\subsection{Control Barrier Function}

In the following, we introduce time varying zeroing CBF (ZCBF) \cite{xu2018constrained} and finite time convergence CBF (FCBF).

\begin{definition}[Time Varying Zeroing CBF (ZCBF) \cite{xu2018constrained}]\label{def:ZCBF}
Consider a dynamical system \eqref{eq:dynamic} and a continuously differentiable function $h:\mathcal{X}\times[0,\infty)\mapsto\mathbb{R}$. If there exists a locally Lipschitz extended class $\mathcal{K}$ function $\alpha$ such that for all $x(t)\in\mathcal{X}$ the following inequality holds
\begin{multline}\label{eq:ZCBF}
    \sup_{u\in\mathcal{U}}\bigg\{\frac{\partial h(x,t)}{\partial x}f(x)+\frac{\partial h(x,t)}{\partial x}g(x)u+\frac{\partial h(x,t)}{\partial t}\\
    +\alpha(h(x,t))\bigg\}\geq 0,
\end{multline}
then function $h$ is a ZCBF.
\end{definition}

Given a ZCBF $h$, the set of controllers satisfying \eqref{eq:ZCBF} is represented as $\mathcal{U}_{Z}(x,t)=\{\mu|\frac{\partial h(x,t)}{\partial x}f(x)+\frac{\partial h(x,t)}{\partial x}g(x)\mu(x,t)\\+\frac{\partial h(x,t)}{\partial t}+\alpha(h(x,t))\geq 0\}$. The following proposition \cite{xu2018constrained} characterizes $\mathcal{U}_{Z}(x,t)$.
\begin{proposition}\label{prop:ZCBF}
Let $\mathcal{C}(t)=\{x|h(x,t)\geq0\}$, where $h:\mathcal{X}\times[0,\infty)\mapsto\mathbb{R}$. Consider a feedback controller $\mu(x,t)\in\mathcal{U}_Z(x,t)$. If $h$ is a ZCBF, then for all $x\in\mathcal{C}(t)$ and $t\geq0$, $\mu(x,t)$ guarantees the set $\mathcal{C}(t)$ to be forward invariant.
\end{proposition}

Motivated by finite time convergence CBF in \cite{li2018formally}, we define time varying FCBF as follows.

\begin{definition}[Time Varying Finite Time Convergence CBF (FCBF)]\label{def:FCBF}
Consider a dynamical system \eqref{eq:dynamic} and a continuously differentiable function $h:\mathcal{X}\times[0,\infty)\mapsto\mathbb{R}$. If there exist constants $\rho\in[0,1)$ and $\gamma>0$ such that for all $x(t)\in\mathcal{X}$ the following inequality holds
\begin{multline}\label{eq:FCBF}
    \sup_{u\in\mathcal{U}}\bigg\{\frac{\partial h(x,t)}{\partial x}f(x)+\frac{\partial h(x,t)}{\partial x}g(x)u\\
    +\frac{\partial h(x,t)}{\partial t}+\gamma\cdot \text{sgn}(h(x,t))|h(x,t)|^\rho\bigg\}\geq 0,
\end{multline}
then function $h$ is a FCBF.
\end{definition}

Given an FCBF $h$, the set of controllers that satisfy \eqref{eq:FCBF} is represented as $\mathcal{U}_F(x,t)=\{\mu|\frac{\partial h(x,t)}{\partial x}f(x)+\frac{\partial h(x,t)}{\partial x}g(x)\mu(x,t)+\frac{\partial h(x,t)}{\partial t}+\gamma\cdot \text{sgn}(h(x,t))|h(x,t)|^\rho\geq 0\}$. The following proposition extends the result in \cite{li2018formally} on time invariant FCBF and characterizes $\mathcal{U}_{F}(x,t)$.
\begin{proposition}\label{prop:FCBF}
Let $\mathcal{C}(t)=\{x|h(x,t)\geq0\}$, where $h:\mathcal{X}\times[0,\infty)\mapsto\mathbb{R}$. Consider a feedback controller $\mu(x,t)\in\mathcal{U}_F(x,t)$. If $h$ is an FCBF, then for any initial state $x_0\in\mathcal{X}$, controller $\mu(x,t)$ guarantees that the system will be steered to the set $\mathcal{C}(t)$ within finite time $0<T<\infty$ such that $x(T)\in\mathcal{C}(T)$. The convergence time $T=\frac{|h(x_0,0)|^{(1-\rho)}}{\gamma(1-\rho)}$. Moreover, the system remains in $\mathcal{C}(t')$ for all $t'\geq T$.

\end{proposition}
\begin{proof}
The proof follows \cite{li2018formally}. Construct a Lyapunov function $V(x,t)=\max\{-h(x,t),0\}$. We can verify that $V(x,t)=0$ for all $x(t)\in\mathcal{C}(t)$, $V(x,t)>0$ for all $x(t)\in\mathcal{X}\setminus\mathcal{C}(t)$, and $\frac{\text{d}}{\text{d}t}V(x,t)\leq\gamma V(x,t)^\rho$. By Theorem 4.1 in \cite{haddad2008finite}, finite time stability holds for system \eqref{eq:dynamic}. Thus if $x_0\in\mathcal{C}(t)$, controllers in $\mathcal{U}_F(x,t)$ render $\mathcal{C}(t)$ forward invariant. If $x_0\notin\mathcal{C}(t)$, then the system converges to $\mathcal{C}(t)$ within finite time $T=\frac{|h(x_0,0)|^{(1-\rho)}}{\gamma(1-\rho)}$.
\end{proof}

\subsection{Design of Control Barrier Functions}\label{sec:CBF}

In this subsection, we first construct a sequence of LTL formulae so that satisfying all formulae is equivalent to satisfying the given LTL formula $\varphi$. Then we show how to design time varying CBFs for each formula.

Given an LTL specification $\varphi$, we compute the DRA associated with $\varphi$, and pick an accepting run $\eta=q_0,\ldots,q_J,(q_{J+1},\ldots,q_{J+N})^\omega$ of the DRA. The complexity of constructing the DRA is doubly exponential in the size of the formula in the worst-case as in the existing works on both abstraction-free and abstraction-based LTL synthesis \cite{ding2014optimal,horowitz2014compositional,wongpiromsarn2015automata}. However, we note that several important classes of LTL formulas have DRAs of polynomial size \cite{babiak2013effective}, and that our approach mitigates the exponential complexity of computing a finite state abstraction.

Rewriting $\eta$ into prefix-suffix form, we have that the sequence of states $q_0,\ldots,q_J$ forms the prefix $\eta_{pref}$, and the sequence of states $q_{J+1},\ldots,q_{J+N}$ forms the suffix $\eta_{suff}$. We denote the transition from state $q_j$ to $q_{j+1}$ as $\eta_j$, and denote the input word of transition $\eta_j$ as $\phi_j$. The input word $\phi_j$ is in the form of conjunction or disjunction of atomic propositions \cite{baier2008principles}, i.e., $\phi_j=\pi_1\bowtie\ldots\bowtie\pi_k$ where $\pi_i\in\Pi$ is an atomic proposition for all $i=1,\ldots,k$ and $\bowtie\in\{\land,\lor\}$.

Given the accepting run $\eta$, we construct a sequence of formulae $\{\psi_j|j=0,1,\ldots,J+N\}$, where $\psi_j$ corresponds to transition $\eta_j$ as follows. We denote the input word corresponding to the self-transition at state $q_j$ as $\Phi_j$, i.e., $\delta(q_j,\Phi_j)=q_j$. We also note that $\phi_j$ is the input word corresponding to a transition from $q_j$ to $q_{j+1}$. We then construct a formula $\psi_j$ corresponding to transition $\eta_j$ as 
\begin{equation}\label{eq:sub-formula}
    \psi_j=\Phi_j~ \mathbf{U}~\Box\left(\phi_j \wedge \Phi_{j+1}\right).
\end{equation}
Formula $\psi_j$ indicates that no transition starting from state $q$ should occur except self-transition and $\eta_j$. Since both prefix and suffix of $\eta$ are over finite horizon, only a finite number of formulae $\{\psi_j|j=0,1,\ldots,J+N\}$ are generated. 

We then assign a sequence of time instants $t_1<\ldots<t_J$ as the deadlines of each transition $\eta_0,\eta_1,\ldots,\eta_J$ of $\eta_{pref}$. The deadlines of the transitions of the suffix can be generated as $n\Delta+t_{J+1}$, where $n$ is a nonnegative integer and $\Delta\geq0$. We additionally let $t_0=0<t_1$. 
Given the sequence of deadlines, we define the active time of each formula $\psi_j$ as $[t_j,t_{j+1}]$, during which formula $\psi_j$ must be satisfied. There are two advantages of defining the active time of each formula $\psi_j$. First, although each formula $\psi_j$ needs to be interpreted over infinite runs, the active time enables us to interpret each $\psi_j$ over finite runs. That is, formula $\psi_j$ needs to be satisfied during $[t_j,t_{j+1}]$. For time $t>t_{j+1}$, formula $\psi_j$ can be violated. Second, the active time allows our approach to satisfy multiple, sequential constraints (e.g., reaching disjoint regions A and B) that cannot be satisfied simultaneously.

Given a time interval $[t,t']$ and controller $\mu$, we let $\mathbf{x}^\mu([t,t'])$ be the system trajectory during time interval $[t,t']$ under controller $\mu$. The trace of $\mathbf{x}^\mu([t,t'])$ is denoted as $Trace\left(\mathbf{x}^\mu([t,t'])\right)$. Denote the system state under controller $\mu$ at time $t$ as $\mathbf{x}^\mu(t)$. We then show the effectiveness of \eqref{eq:sub-formula} by analyzing the relationship between satisfying each formula $\psi_j$ and run $\eta$.

\begin{lemma}\label{lemma:transition}
    Let $\eta$ be an accepting run and $\psi_j$ be a formula in the form of \eqref{eq:sub-formula} whose active time is $[t_j,t_{j+1}]$. If $Trace\left(\mathbf{x}^\mu([0,t_j])\right)$ steers the DRA from $q_0$ to $q_j$, and $Trace\left(\mathbf{x}^\mu([t_j,t_{j+1}])\right)\models\psi_j$, then the DRA transitions from state $q_j$ to $q_{j+1}$ during time interval $[t_j,t_{j+1}]$. Moreover, the DRA remains in state $q_{j+1}$ until at least time $t_{j+1}$.
\end{lemma}
\begin{proof}
    We first prove that if $Trace\left(\mathbf{x}^\mu([0,t_j])\right)$ steers the DRA from $q_0$ to $q_j$, and $Trace\left(\mathbf{x}^\mu([t_j,t_{j+1}])\right)\models\psi_j$, then the DRA transitions from state $q_j$ to $q_{j+1}$ during time interval $[t_j,t_{j+1}]$. We prove by contradiction. Suppose the current state of the DRA is $q_j$ and $Trace\left(\mathbf{x}^\mu([t_j,t_{j+1}])\right)\models\psi_j$, while the DRA transitions from state $q_j$ to some state $q'\neq q_{j+1}$. By the semantics of until operator $\mathbf{U}$, there must exist some time $t\in[t_j,t_{j+1}]$ such that $L(\mathbf{x}^\mu(t'))\models\phi_j\land\Phi_{j+1}$ for all $t'\in[t,t_{j+1}]$ in order to make $Trace\left(\mathbf{x}^\mu([t_j,t_{j+1}])\right)\models\psi_j$ hold. Then by the semantics of and operator $\land$, $L(\mathbf{x}^\mu(t))\models\phi_j\land\Phi_{j+1}$ implies that $L(\mathbf{x}^\mu(t))\models\phi_j$. Since $\phi_j$ is the input associated with the transition from $q_j$ to $q_{j+1}$, then $q'=q_{j+1}$. Otherwise, the DRA contains nondeterminism which conflicts Definition \ref{def:DRA}.
    
    We then prove that the DRA remains in $q_{j+1}$ until at least time $t_{j+1}$. Suppose the DRA transitions from $q_{j+1}$ to some state $q$ before $t_{j+1}$. This is equivalent to the fact that there exist some state $q\in\mathcal{N}(q_{j+1})$ and time $t\in[t_j,t_{j+1}]$ such that $L(\mathbf{x}^\mu(t))\models\phi_{j+1}^q$, where $\phi_{j+1}^q$ is the input word associated with transition from state $q_{j+1}$ to some neighbor state $q$. However, this contradicts $\Phi_{j+1}$, and thus the DRA cannot transition to some state $q\in\mathcal{N}(q_{j+1})$. By the definition of neighbor states $\mathcal{N}(q_{j+1})$, the DRA can only take the self-transition at $q_{j+1}$.
\end{proof}
Inducting the results on Lemma \ref{lemma:transition} gives the following result.
\begin{corollary}\label{coro:run}
    If $Trace\left(\mathbf{x}^\mu([t_j,t_{j+1}])\right)\models\psi_j$ for all $j=0,1,\ldots$, then $Trace\left(\mathbf{x}^\mu\right)\models\varphi$.
\end{corollary}

Lemma \ref{lemma:transition} and Corollary \ref{coro:run} imply that, in order to ensure that the specification is satisfied, it suffices to ensure that the trajectory under controller $\mu$ satisfies each $\psi_{j}$ within its active time $[t_{j},t_{j+1}]$. In what follows, we construct a set of CBFs that will be used to ensure satisfaction of each $\psi_{j}$.

In the following, we design CBFs for $\Phi_j$ and $\phi_j\land\Phi_{j+1}$. We first define a CBF $h_\pi$ for each atomic proposition $\pi$ that is involved in $\psi_j$. We consider CBFs in the form of $h_{\pi}(x,t)=M_{\pi}(t)+Z_{\pi}(x)$ for all $\pi$, where $M_{\pi}(t)$ and $Z_{\pi}(x)$ are called guard function and state function, respectively. The state function $Z_{\pi}(x)$ is a function of state $x$ that captures if the state $x$ is in $\llbracket\pi\rrbracket$, i.e., $\llbracket\pi\rrbracket=\{x|Z_\pi(x)\geq0\}$. The guard function $M_{\pi}(t)=\frac{E_{\pi}}{1+e^{-b_{\pi}(t+c_{\pi})}}-\epsilon_{\pi}$ is a logistic function, where $E_{\pi}>0$, $b_{\pi}>0$, and $\epsilon_{\pi}\geq 0$. The guard function $M_{\pi}(t)$ is introduced so that each atomic proposition $\pi$, and hence $\psi_j$, only need to be satisfied during their active time. 

We then show how to choose $E_{\pi},b_{\pi},c_{\pi}$, and $\epsilon_{\pi}$ for each $\pi$. First, if atomic proposition $\pi$ is satisfied at time $t=0$, then $h_{\pi}(x_0,0)\geq0$. If  $\pi$ is not satisfied at time $t=0$, then $h_{\pi}(x_0,0)<0$. These two requirements are captured by \eqref{eq:initial constraint 1} and \eqref{eq:initial constraint 2}. Second, given the deadline $t_j$ of atomic proposition $\pi$, we have $M_{\pi}(t_j)\leq 0$, as shown in \eqref{eq:negative constraint}. To summarize, we have the following inequalities:
\begin{subequations}\label{eq:coefficient}
\begin{align}
    & \frac{E_{\pi}}{1+e^{-b_{\pi}c_{\pi}}}-\epsilon_{\pi}+Z_{\pi}(x_0)\geq 0,\mbox{ if }\pi\in L(x_0)\label{eq:initial constraint 1}\\
    &\frac{E_{\pi}}{1+e^{-b_{\pi}c_{\pi}}}-\epsilon_{\pi}+Z_{\pi}(x_0)< 0,\mbox{ if }\pi\notin L(x_0)\label{eq:initial constraint 2}\\
    &\frac{E_{\pi}}{1+e^{-b_{\pi}(t_j+c_{\pi})}}-\epsilon_{\pi}\leq 0,\label{eq:negative constraint}\\
    &E_{\pi}>0,b_{\pi}>0,\epsilon_{\pi}\geq0\label{eq:para constraint}
\end{align}
\end{subequations}

Inequalities \eqref{eq:coefficient} are solved as follows. We first pick some $b_{\pi}>0$ and $c_{\pi}$ such that $c_{\pi}\leq t_{j+1}$ if $\pi$ is involved in $\Phi_j$ or $\phi_j$, and $c_{\pi}\leq t_{j+2}$ if $\pi$ is involved in $\Phi_{j+1}$. Fixing the values of $b_{\pi}$ and $c_{\pi}$, then $E_{\pi}$ and $\epsilon_{\pi}$ can be obtained by solving the linear inequalities \eqref{eq:coefficient}. 
We characterize the CBFs obtained by solving \eqref{eq:coefficient} using the following lemma.
\begin{lemma}\label{lemma:CBF}
    Let $t_j$ be the deadline of atomic proposition $\pi$, and $h_\pi$ be the CBF obtained by solving \eqref{eq:coefficient}. For any $t\leq t_j$, if $h_\pi(x,t)\geq0$, then $x(t)\in\llbracket\pi\rrbracket$.
\end{lemma}
\begin{proof}
    Inequality \eqref{eq:negative constraint} indicates that $M_{\pi}(t_j)\leq 0$, where $t_j$ is the deadline of ${\pi}$. By \eqref{eq:para constraint}, the guard function $M_{\pi}(t)$ is monotone increasing. Therefore, for all $t\leq t_j$, $M_{\pi}(t)\leq 0$ holds. By the definition of $h_{\pi}(x,t)$, we have that $Z_{\pi}(x)= h_{\pi}(x,t)-M_{\pi}(t)$. Given $M_{\pi}(t)\leq0$ for all $t\leq t_j$ and $h_{\pi}(x,t)\geq0$, we have $Z_{\pi}(x)\geq0$, and thus the system state $x(t)$ is in the region $\{x(t)|Z_{\pi}(x(t))\geq 0\}=\llbracket\pi\rrbracket$.
\end{proof}

Given a CBF $h_\pi$ for each atomic proposition $\pi$ that is involved in $\psi_j$, we compute the CBFs for $\Phi_j$ and $\phi_j\land\Phi_{j+1}$. We note that $\Phi_j$ and $\phi_j\land\Phi_{j+1}$ are both in the forms of conjunctions/disjunctions of atomic propositions \cite{baier2008principles}. We utilize the following definition to construct the CBF for $\Phi_j$ and $\phi_j\land\Phi_{j+1}$.
\begin{definition}\label{def:CBF construction}
    Consider a set of atomic proposition $\{\pi_i|i=1,\ldots,k\}$. Let $h_{\pi_i}:\mathbb{R}^n\times [0,\infty)\mapsto\mathbb{R}$ be the CBF of each $\pi_i$ defined as $h_{\pi_i}(x,t)=M_{\pi_i}(t)+Z_{\pi_i}(x)$ for each atomic proposition $\pi_i$, where $M_{\pi_i}(t)$ is computed by \eqref{eq:coefficient} and $Z_{\pi_i}(t)$ is defined as $\{x|Z_{\pi_i}(x)\geq0\}=\llbracket\pi_i\rrbracket$. Consider a formula $\phi'=\pi_1\bowtie\ldots\bowtie\pi_{k-1}$, where $\bowtie\in\{\land,\lor\}$. Let $h_{\phi'}$ be the CBF of $\phi'$. Then the CBF of formula $\phi=\phi'\land\pi_k$ is
\begin{equation}\label{eq:min approximation}
    h_\phi(x,t)=-\ln\left[\exp{(-h_{\phi'}(x,t))}+\exp{(-h_{\pi_k}(x,t))}\right].
\end{equation}
    The CBF $h_\phi$ of formula $\phi=\phi'\lor\pi_k$ for some $\lambda>0$ is
\begin{equation}\label{eq:max approximation}
    h_\phi(x,t)=\frac{h_{\phi'}(x,t)e^{\lambda h_{\phi'}(x,t)}+h_{\pi_k}(x,t)e^{\lambda h_{\pi_k}(x,t)}}{e^{\lambda h_{\phi'}(x,t)}+e^{\lambda h_{\pi_k}(x,t)}}.
\end{equation}
\end{definition}
Definition \ref{def:CBF construction} recursively defines the CBF for a formula $\phi$ in the form of $\phi=\pi_1\bowtie\ldots\bowtie\pi_{k}$, where $\bowtie\in\{\land,\lor\}$. By Lemma \ref{lemma:min operator}, we have that CBFs \eqref{eq:min approximation} and \eqref{eq:max approximation} bound $\min\{h_{\phi'}(x,t),h_{\pi_k}(x,t)\}$ and $\max\{h_{\phi'}(x,t),h_{\pi_k}(x,t)\}$ from below, respectively. When $\Phi_j$ and $\phi_j\land\Phi_{j+1}$ are in the forms of $\pi_1\bowtie\ldots\bowtie\pi_{k}$, where $\bowtie\in\{\land,\lor\}$, their CBFs can be obtained by recursively applying Definition \ref{def:CBF construction}.

\begin{center}
  	\begin{algorithm}[!htp]
  		\caption{Algorithm for computing the CBFs for each formula $\psi_j$.}
  		\label{alg:CBF design}
  		\begin{algorithmic}[1]
  			\Procedure{CBF\_Design}{$\varphi$}
  			\State \textbf{Input}: LTL specification $\varphi$
  			\State \textbf{Output:} CBFs for each formula $\psi_j$
  			\State Compute the DRA associated with LTL specification $\varphi$, and the set of accepting runs on the DRA.
  			\State Pick an accepting run $\eta$ on the DRA, and identify each formula $\psi_j$ associated with each transition $\eta_j$ in $\eta$ as \eqref{eq:sub-formula}.
  			\State Specify a sequence of time $0<t_1<\ldots$ for accepting run $\eta$. 
  	        \State Pick a set of feasible coefficients for relations \eqref{eq:coefficient} for each atomic proposition $\pi$ involved in $\psi_j$.
  	        \State Recursively compute CBFs for $\Phi_j$ and $\phi_j\land\Phi_{j+1}$ using Definition \ref{def:CBF construction}.
  	        \State\Return CBFs for $\Phi_j$ and $\phi_j\land\Phi_{j+1}$
  			\EndProcedure        	
  		\end{algorithmic}
    \end{algorithm}
\end{center}

The procedure we used to design the CBFs for each $\psi_j$ is presented in Algorithm \ref{alg:CBF design}. We characterize the construction of CBFs for $\Phi_j$ and $\phi_j\land\Phi_{j+1}$ using the following proposition.
\begin{lemma}\label{lemma:compose CBF}
    Let $h_\phi$ be the CBF obtained by Algorithm \ref{alg:CBF design}, where $\phi\in\{\Phi_j,\phi_j\land\Phi_{j+1}\}$. For any time $t\in[t_j,t_{j+1}]$, if $h_\phi(x,t)\geq0$, then $L(x(t))\models\phi$.
\end{lemma}
\begin{proof}
    We prove by induction. Consider $\phi=\pi_1\land\pi_2$. Then $h_\phi$ is computed as \eqref{eq:min approximation}. By Lemma \ref{lemma:min operator}, $h_\phi(x,t)\geq0$ implies that $h_{\pi_1}(x,t)\geq0$ and $h_{\pi_2}(x,t)\geq0$. By Lemma \ref{lemma:CBF}, $x(t)\in\llbracket\pi_1\rrbracket\cap\llbracket\pi_2\rrbracket$, and thus $\phi$ is satisfied. Consider $\phi=\pi_1\lor\pi_2$. Then $h_\phi$ is computed as \eqref{eq:max approximation}. By Lemma \ref{lemma:min operator}, $h_\phi(x,t)\geq0$ implies that $h_{\pi_1}(x,t)\geq0$ or $h_{\pi_2}(x,t)\geq0$. By Lemma \ref{lemma:CBF}, $x(t)\in\llbracket\pi_1\rrbracket\cup\llbracket\pi_2\rrbracket$, and thus $\phi$ is satisfied. These two cases serve as our induction base.
    
    Suppose the lemma holds after applying \eqref{eq:min approximation} and \eqref{eq:max approximation} $k-1$ times, and denote the corresponding CBF as $h_\phi^{k-1}$. If the $k$-th operation is a conjunction with atomic proposition $\pi$, then $h_\phi^{k}(x,t)$ is obtained by \eqref{eq:min approximation}. By Lemma \ref{lemma:min operator}, $h_\phi^{k}(x,t)\geq0$ implies that $h_{\pi}(x,t)\geq0$ and $h_{\phi}^{k-1}(x,t)\geq0$. By Lemma \ref{lemma:CBF}, $h_{\pi}(x,t)\geq0$ implies $x(t)\in\llbracket\pi\rrbracket$. By our inductive hypothesis, $h_{\phi}^{k-1}(x,t)\geq0$ implies $L(x(t))\models\phi$ after applying \eqref{eq:min approximation} and \eqref{eq:max approximation} $k-1$ times. Combine the arguments above, we have $L(x(t))\models\phi$ when the $k$-th operator is a conjunction. If the $k$-th operation is a disjunction with atomic proposition $\pi$, then $h_\phi^{k}(x,t)$ is obtained by \eqref{eq:max approximation}. By Lemma \ref{lemma:min operator}, $h_\phi^{k}(x,t)\geq0$ implies that $h_{\pi}(x,t)\geq0$ or $h_{\phi}^{k-1}(x,t)\geq0$. Similar to the conjunction case, we can conclude $L(x(t))\models\phi$ when the $k$-th operator is a disjunction. Since $\phi$ is in the form of conjunction and disjunction of finite number of atomic propositions, we have that the lemma holds by induction.
\end{proof}

\subsection{CBF-based Controller Synthesis}

In this subsection, we present how to compute the controllers that satisfy each formula $\psi_j$ in the form of \eqref{eq:sub-formula}. 
We then show that by satisfying all formulae $\psi_j$ for all $j$, the LTL specification $\varphi$ is satisfied.

\begin{lemma}\label{lemma:compose atomic until}
    Consider a formula $\psi_j$ in the form of \eqref{eq:sub-formula} whose active time is $[t_j,t_{j+1}]$. Let $h_{\Phi_j}$ and $h_{\Omega_j}$ be the CBFs of $\Phi_j$ and $\phi_j\land\Phi_{j+1}$ obtained using Algorithm \ref{alg:CBF design}, respectively. Then for all $t\in[t_j,t_{j+1}]$ any feedback controller in 
    \begin{equation*}
    \mathcal{U}_{\psi_j}(x,t)=\begin{cases}
    \mathcal{U}_2(x,t),\quad\quad\quad\quad\mbox{ if }L(x(t))\models\phi_j\land\Phi_{j+1},\\
    \mathcal{U}_1(x,t)\cap\mathcal{U}_2(x,t),\mbox{ if }L(x(t))\models\Phi_j,\\
    \emptyset,\quad\quad\quad\quad\quad\quad\mbox{ otherwise}
    \end{cases}
    \end{equation*}
    satisfies $Trace\left(\mathbf{x}^\mu([t_j,t_{j+1}])\right)\models\psi_j$, where
    \begin{align*}
    &\mathcal{U}_1(x,t)=\Big\{\mu|\frac{\partial h_{\Phi_j\Omega_j}(x,t)}{\partial x}f(x)+\frac{\partial h_{\Phi_j\Omega_j}(x,t)}{\partial x}g(x)\mu(x,t)\\
    &\quad\quad\quad\quad+\frac{\partial h_{\Phi_j\Omega_j}(x,t)}{\partial t}+\alpha(h_{\Phi_j\Omega_j}(x,t))\geq 0\Big\},\\
    &\mathcal{U}_2(x,t)=\Big\{\mu|\frac{\partial h_{\Omega_j}(x,t)}{\partial x}f(x)+\frac{\partial h_{\Omega_j}(x,t)}{\partial x}g(x)\mu(x,t)\\
    &\quad+\frac{\partial h_{\Omega_j}(x,t)}{\partial t}+\gamma\cdot \text{sgn}(h_{\Omega_j}(x,t))|h_{\Omega_j}(x,t)|^\rho\geq 0\Big\}\\
    &h_{\Phi_j\Omega_j}(x,t)=\frac{h_{\Phi_j}(x,t)e^{\lambda h_{\Phi_j}(x,t)}+h_{\Omega_j}(x,t)e^{\lambda h_{\Omega_j}(x,t)}}{e^{\lambda h_{\Phi_j}(x,t)}+e^{\lambda h_{\Omega_j}(x,t)}}.
\end{align*}
\end{lemma}
\begin{proof}
 
 By the semantics of until operator $\mathbf{U}$ and Definition \ref{def:trace}, it suffices to show that there exists a time sequence $t_j,T,t_{j+1}$ such that (i) $t_j<t_{j+1}$, (ii) $T\in[t_j, t_{j+1}]$, and (iii) $L(x(t))\models \Phi_j$ for all $t\in[t_j,T)$ and $L(x(t))\models \phi_j\land\Phi_{j+1}$ for all $t\in[T,t_{j+1}]$. We show that each of these conditions is satisfied.
 
 First, condition $t_j<t_{j+1}$ holds by the construction of deadlines as given in Section \ref{sec:CBF}. To guarantee $T\in[t_j,t_{j+1}]$, we can tune parameters $\rho$ and $\gamma$ as given in Proposition \ref{prop:FCBF} so that $T\leq t_{j+1}$. In the following, we show that there exists some time $T\in[t_j,t_{j+1}]$ such that $L(x(t))\models \Phi_j$ for all $t\in[t_j,T)$ and $L(x(t))\models \phi_j\land\Phi_{j+1}$ for all $t\in[T,t_{j+1}]$.
 
 We start with the case where $L(x(t))\models\phi_j\land\Phi_j$. In this case, $T=t$ and we need to guarantee $L(x(t'))\models \phi_j\land\Phi_{j+1}$ for all $t'\in[t,t_{j+1}]$ so that formula $\psi_j$ is satisfied. By Proposition \ref{prop:FCBF}, the set of controllers $\mathcal{U}_2(x,t)$ ensures that the system remains in the set $\{x|h_{\Omega_j}(x,t')\geq0\}$ for all $t'\geq t$. By Lemma \ref{lemma:compose CBF}, we have that $L(x(t'))\models\phi_j\land\Phi_j$ for all $x(t')\in\{x(t)|h_{\Omega_j}(x,t')\geq0\}$ for all $t'\in [t,t_{j+1}]$. Therefore, $Trace\left(\mathbf{x}^\mu([t,t_{j+1}])\right)\models\psi_j$ for all $t\in[t_j,t_{j+1}]$ in this case.
 
 We then consider the case where $L(x(t))\models\Phi_j$. By Proposition \ref{prop:ZCBF}, the set of controllers $\mathcal{U}_1(x,t)$ guarantees that the system remains in the set $\{x(t')|h_{\Phi_j\Omega_j}(x,t')\geq0\}$ for all $t'\in[t,t_{t+1}]$. By Lemma \ref{lemma:compose CBF}, we have that $L(x(t'))\models\Phi_j$, or $L(x(t'))\models\phi_j\land\Phi_{j+1}$, or both, for all $x(t')\in\{x(t')|h_{\Phi_j\Omega_j}(x,t')\geq0\}$ where $t'\in[t,t_{j+1}]$. By Proposition \ref{prop:FCBF}, the set of controllers $\mathcal{U}_2(x,t)$ ensures that the system will be steered to the set $\{x|h_{\Omega_j}(x,t)\geq0\}$ at some time $T\geq t$ if $x(t)\notin \{x|h_{\Omega_j}(x,t)\geq0\}$. Moreover, the system remains in $x(t')\in\{x|h_{\Omega_j}(x,t')\geq0\}$ for all $t'\in[T,t_{j+1}]$. By Lemma \ref{lemma:compose CBF}, we have that $Trace\left(\mathbf{x}^\mu([t,t_{j+1}])\right)\models\Box( \phi_j\land\Phi_{j+1})$ for all $t\in[t_j,t_{j+1}]$. Thus, the controllers in $\mathcal{U}_1(x,t)\cap\mathcal{U}_2(x,t)$ guarantee $Trace\left(\mathbf{x}^\mu([t,t_{j+1}])\right)\models\psi_j$ for all $t\in[t_j,t_{j+1}]$.
 
 Combining the arguments above, we have that the controllers in $\mathcal{U}_{\psi_j}(x,t)$ satisfy $Trace\left(\mathbf{x}^\mu([t_j,t_{j+1}])\right)\models\psi_j$.
\end{proof}

We note that the computation of controllers can be simplified when the formula $\psi_j$ is in some simple forms. Consider a formula $\psi_j$ as per \eqref{eq:sub-formula}. When $\Phi_j=\top$, then $\psi_j=\Box(\phi_j\land\Phi_{j+1})$. The controllers satisfying $\psi_j$ in this case are given by the following corollary.
\begin{corollary}\label{prop:compose atomic box}
    Suppose $\psi_j=\Box\phi$. Let $h_{\phi}$ be the CBF of $\phi$ that is obtained using Algorithm \ref{alg:CBF design}. Then any feedback controller in
    \begin{equation*}
    \mathcal{U}_{\psi_j}(x,t)=\begin{cases}
    \big\{\mu|\frac{\partial h_{\phi}(x,t)}{\partial x}f(x)+\frac{\partial h_{\phi}(x,t)}{\partial x}g(x)\mu(x,t)+\\
    \frac{\partial h_{\phi}(x,t)}{\partial t}+\alpha(h_{\phi}(x,t))\geq 0\big\},\mbox{ if } L(x_0)\models\phi\\
    \emptyset,\mbox{ otherwise}
    \end{cases}
    \end{equation*}
    satisfies $\psi_j$.
\end{corollary}
\begin{proof}
    The proof follows by replacing $\Phi_j$ in Lemma \ref{lemma:compose atomic until} with unconditionally true $\top$.
\end{proof}

We finally show that by satisfying each formula $\psi_j$, specification $\varphi$ is satisfied.

\begin{theorem}\label{thm:correctness}
Applying the controllers in $\mathcal{U}_{\psi_j}(x,t)$ for all $t\geq0$ as given in Lemma \ref{lemma:compose atomic until} renders $Trace\left(\mathbf{x}^\mu\right)\models\varphi$.
\end{theorem}
\begin{proof}
     Applying $\mu$ for each time interval $[t_j,t_{j+1}]$, we have that $Trace\left(\mathbf{x}^\mu([t_j,t_{j+1}])\right)\models\psi_j$ due to Lemma \ref{lemma:compose atomic until}. Then according to Lemma \ref{lemma:transition} and Corollary \ref{coro:run}, we have that satisfying the sequence of formulae $\psi_j$ for all $j$ is equivalent to executing the run $\eta$ on the DRA. Since $\eta$ is an accepting run, we can conclude that $Trace\left(\mathbf{x}^\mu\right)\models\varphi$.
\end{proof}

\begin{figure*}[t!]
\centering
                 \begin{subfigure}{.3\textwidth}
                 \includegraphics[width=\columnwidth]{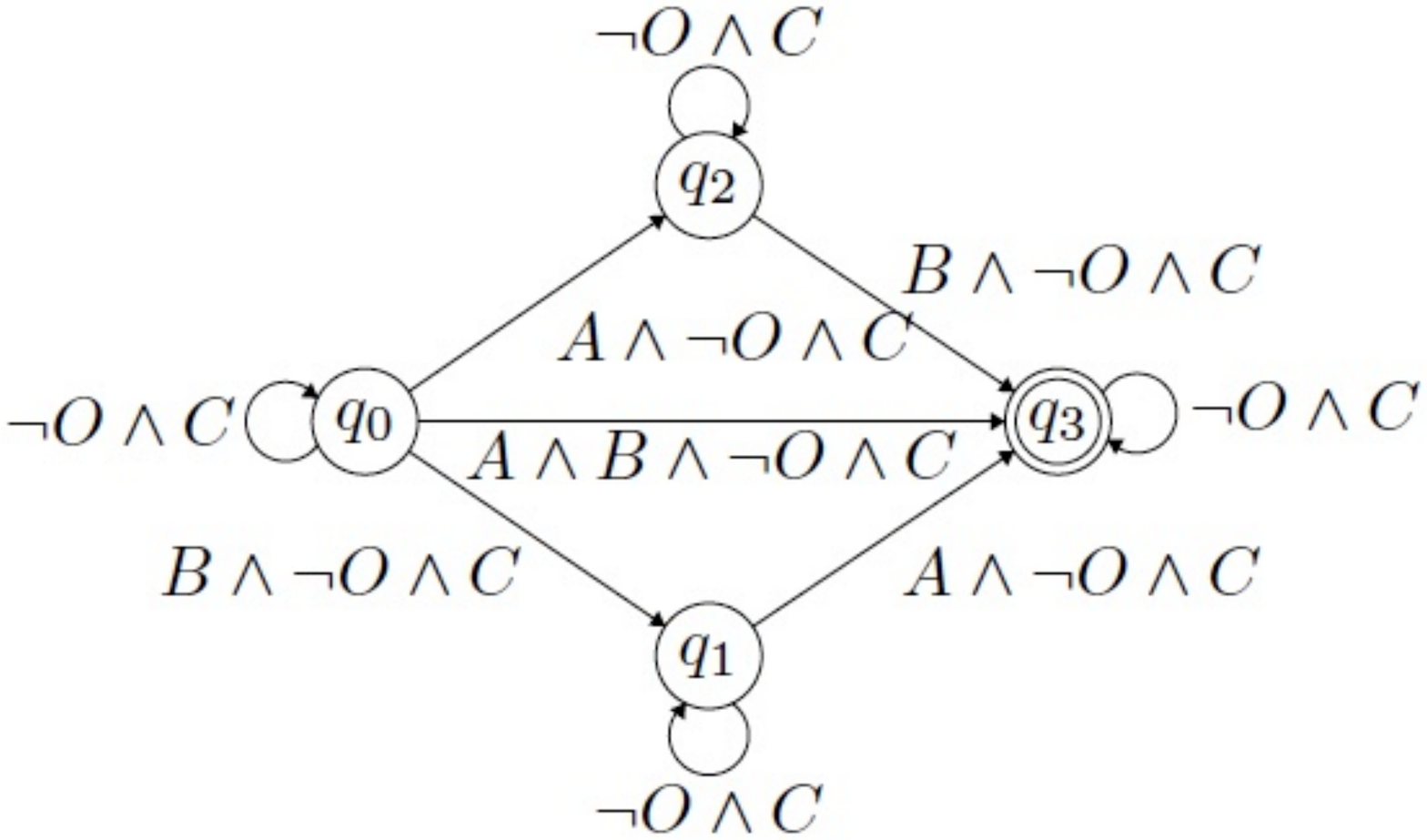}
                 \subcaption {}
                 \label{fig:auto}
                 \end{subfigure}
                 \begin{subfigure}{.3\textwidth}
                 \includegraphics[width=\columnwidth]{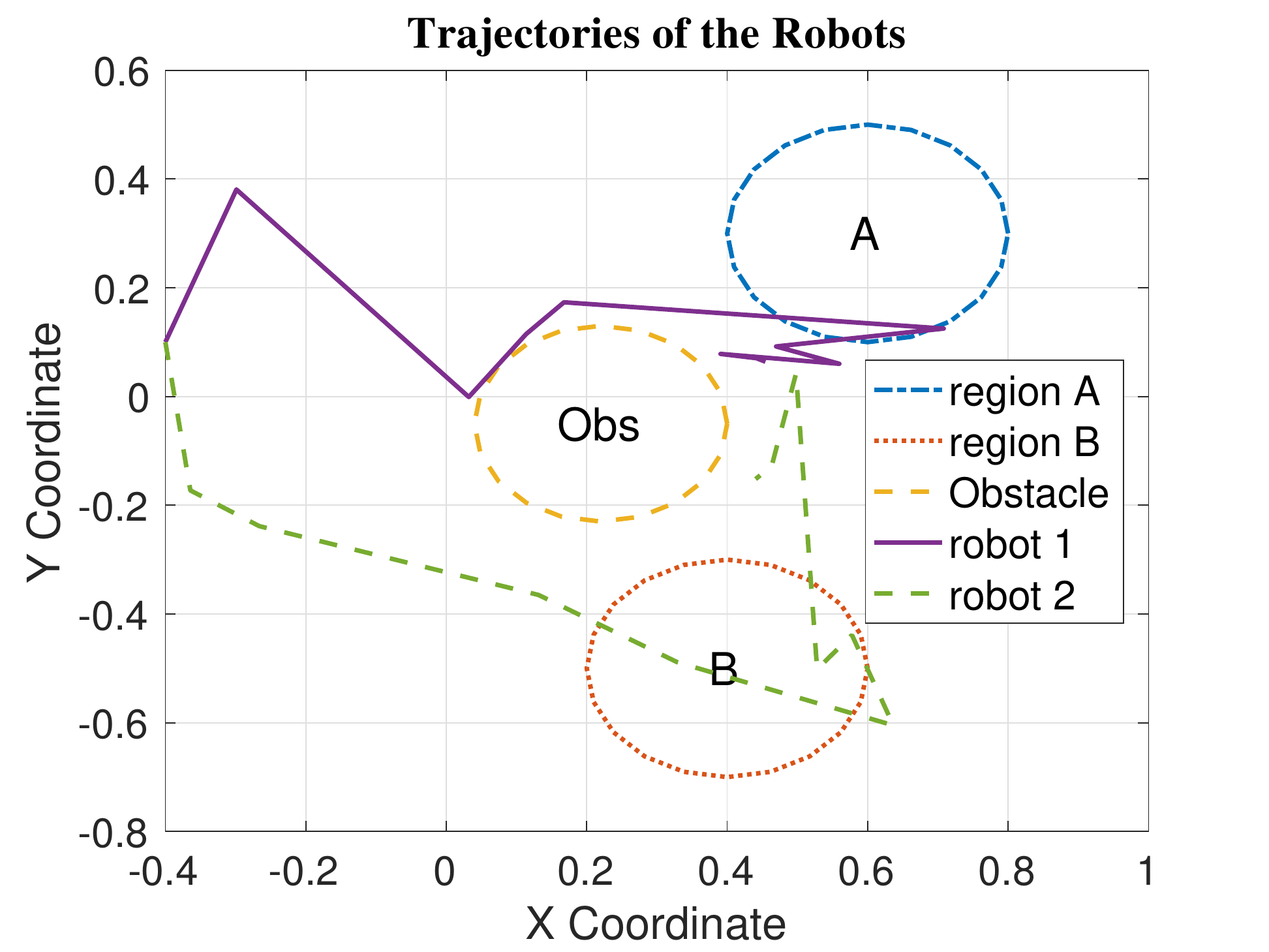}
                 \subcaption {}
                 \label{fig:traj}
                 \end{subfigure}
                 \begin{subfigure}{.3\textwidth}
                 \includegraphics[width=\columnwidth]{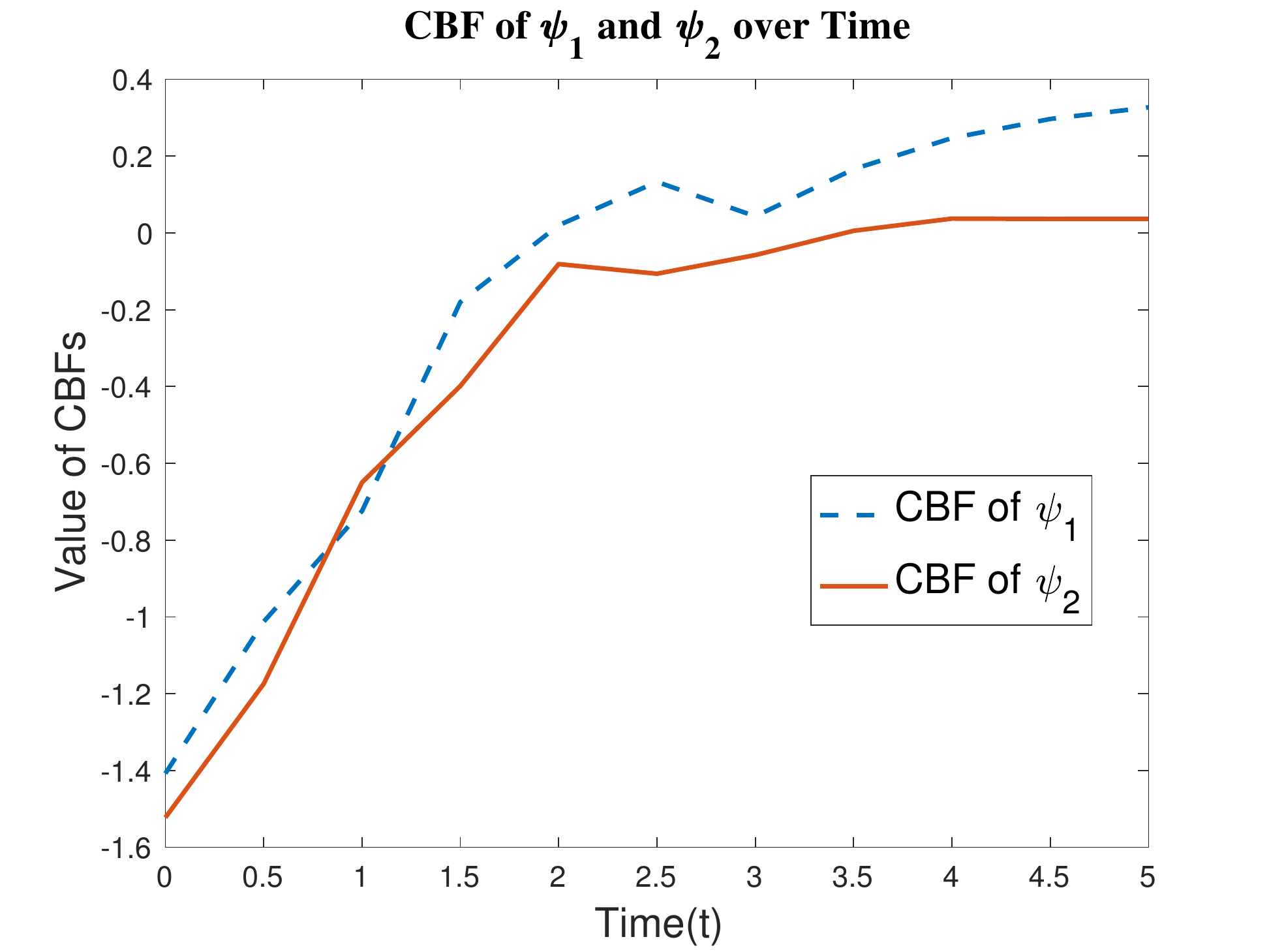}
                 \subcaption {}
                 \label{fig:CBF}
                 \end{subfigure}
\caption{Fig. \ref{fig:auto} shows the DRA associated with $\varphi=\Diamond A\land \Diamond B\land \Box(\neg O\land C)$. It has one Rabin pair $(\emptyset,\{q_3\})$. Fig. \ref{fig:traj} presents the trajectories of both robot using the proposed approach in this paper. The trajectory of the first robot is plotted in solid line, while the trajectory of the second robot is plotted in dotted line. The first robot eventually reaches region $A$, and the second robot eventually reaches region $B$. Both robots successfully avoid the obstacle region $O$ and remain close enough to each other. CBFs $h_{\psi_1}$ and $h_{\psi_2}$ are presented in Fig. \ref{fig:CBF}. $h_{\psi_1}$ is plotted in dotted line and $h_{\psi_2}$ is plotted in solid line. Both CBFs become positive at some time $t\geq0$, at which the formulae corresponding to the CBFs are satisfied.}
\end{figure*}

Using the treatment in \cite{ames2014control,ames2016control}, we can compute the control input $u(t)$ at each time $t$ by solving the quadratic program (QP): 
\begin{align*}
    \min_u \quad&u^TPu\\
    \st \quad&u(t)\in\mathcal{U}_{\psi_j}(x,t),~\forall \psi_j
\end{align*}
where $P\in\mathbb{R}^{m\times m}$ is a positive semi-definite matrix that quantifies the cost of control, and $\mathcal{U}_{\psi_j}(x,t)$ is given by Lemma \ref{lemma:compose atomic until}. The constraint set of this QP requires that the CBF constraints as given in Lemma \ref{lemma:compose atomic until} must be satisfied for each formula $\psi_j$.
We note that the deadlines of formulae $\psi_j$'s corresponding to transitions in $\eta_{suff}$ are generated periodically as $n\Delta+t_{J+1}$. We can impose the constraint $u(t) \in \mathcal{U}_{\psi_{i}}(x,t)$ for $i \in \{j,\ldots,j+J+N\}$, resulting in a finite number of constraints in the QP. The controllers for future time can then be generated by implementing the controllers periodically.
\section{Simulation}\label{sec:simulation}

In this section, we present a numerical case study on a two robot homogeneous multi-agent system. Consider a multi-agent system consisting of two robots whose dynamics are given as $\dot{x}=u$, where $x\in\mathbb{R}^4$ is the system state, and $u\in\mathbb{R}^4$ is the input. The state variables $x_1$ and $x_2$ give the coordinate of the first robot, and state variables $x_3$ and $x_4$ give the coordinate of the second robot. The initial positions for both robots are $[-0.4,0.1]$.

Each of the robots has its respective goal region, denoted as region $A$ and region $B$. Both robots are required to eventually reach their goal regions. In the meantime, both robots need to avoid the obstacle region, denoted as $O$. Furthermore, the robots keep exchanging information with each other, and hence must remain close enough. The tasks can be represented by an LTL formula $\varphi=\Diamond A\land \Diamond B\land \Box(\neg O\land C)$, where $C$ represents the connectivity specification. In this case study, region $A$ is modeled as $\{x|Z_{A}(x)\geq 0\}$ where $Z_A(x)=0.2 - \|[x_1,x_2] - [0.6,0.3]\|_2$. Region $B$ is modeled as $\{x|Z_B(x)\geq 0\}$ where $Z_B(x)=0.2- \|[x_3,x_4] - [0.4,-0.5]\|_2$. The obstacle region is modeled as $\{x|Z_{O}(x)\geq 0\}$ where $Z_{O}(x)=0.18 - \|x-[0.22,-0.05,0.22,-0.05]\|_2$. The connectivity between the two robots is given as $C=\{x|Z_{C}(x)\geq0\}$ where $Z_{C}(x)=\sqrt{x_3+0.39}-\|[x_1,x_2]-[x_3,x_4]\|_2$. 

We compare the proposed approach with the approach proposed in \cite{srinivasan2018control}. Since there is single CBF for each region of interest, the relaxation proposed in \cite{srinivasan2018control} coincides with the traditional CBF-based approach. We observe that no feasible trajectory is synthesized since there is no feasible solution to the QP. The infeasibility of the QP is caused by the fact that the connectivity constraint requires the robots to stay close to each other, while steering the robots into region $A$ and $B$ makes them violate the connectivity constraint. This infeasibility agrees with the example presented in \cite{srinivasan2019control}. 

In the following, we demonstrate our proposed approach. Given the LTL specification $\varphi$, the DRA representing $\varphi$ is shown in Fig. \ref{fig:auto}. The accepting runs of the DRA include $q_0q_1(q_3)^\omega$, $q_0(q_3)^\omega$, and $q_0q_2(q_3)^\omega$. We pick the run $\eta=q_0q_1(q_3)^\omega$ in this case study. Transition from $q_0$ to $q_1$ of the accepting run $\eta$ corresponds to formula $\psi_0=(\neg O\land C)\mathbf{U}\Box(B\land\neg O\land C)$. Transition from $q_1$ to $q_2$ corresponds to formula $\psi_1=(\neg O\land C)\mathbf{U}\Box(A\land\neg O\land C)$. Self transition at $q_3$ corresponds to formula $\psi_3=\Box(\neg O\land C)$. Next, we assign the active time for each formulae during which the formula needs to be satisfied. In this case study, we let $\psi_1$ be satisfied during $[0,2]$, and let $\psi_2$ be satisfied during $[2,4]$. Using \eqref{eq:coefficient}, we then construct the CBFs for atomic propositions $B$ and $A$ as $h_B(x,t)=\frac{1}{1+e^{-(t-1.5)}}-0.63+0.2- \|[x_1,x_2] - [0.4,-0.5]\|_2$, $h_A(x,t)=\frac{1}{1+e^{-(t-0.5)}}-0.9+0.2-\|[x_1,x_2] - [0.6,0.3]\|_2$, respectively. Then the CBF for each formula can be constructed by Definition \ref{def:CBF construction}. For instance, $h_{\psi_3}(x,t) = -\ln(\exp(-Z_{O}(x))+\exp(-Z_{C}(x)))$.

In the following, we formulate the QP to solve for the controllers so that the LTL specification $\varphi$ is satisfied. According to Lemma \ref{lemma:compose atomic until}, we have the following constraints for each formula $\psi_j$:
\begin{align}
    &\frac{\partial h_{\Phi_j\Omega_j}(x,t)}{\partial x}f(x)+\frac{\partial h_{\Phi_j\Omega_j}(x,t)}{\partial x}g(x)\mu(x,t)\nonumber\\
    &\quad\quad\quad\quad+\frac{\partial h_{\Phi_j\Omega_j}(x,t)}{\partial t}+\alpha(h_{\Phi_j\Omega_j}(x,t))\geq 0\label{eq:constraint1}\\
    &\frac{\partial h_{\Omega_j}(x,t)}{\partial x}f(x)+\frac{\partial h_{\Omega_j}(x,t)}{\partial x}g(x)\mu(x,t)+\frac{\partial h_{\Omega_j}(x,t)}{\partial t}\nonumber\\
    &\quad\quad\quad\quad+\gamma\cdot \text{sgn}(h_{\Omega_j}(x,t))|h_{\Omega_j}(x,t)|^\rho\geq 0\label{eq:constraint2}
\end{align}
where $\Phi_j$ and $\Omega_j$ are defined as given in Lemma \ref{lemma:compose atomic until}, and
\begin{equation*}
    h_{\Phi_j\Omega_j}(x,t)=\frac{h_{\Phi_j}(x,t)e^{\lambda h_{\Phi_j}(x,t)}+h_{\Omega_j}(x,t)e^{\lambda h_{\Omega_j}(x,t)}}{e^{\lambda h_{\Phi_j}(x,t)}+e^{\lambda h_{\Omega_j}(x,t)}}.
\end{equation*}
Formulating \eqref{eq:constraint1} and \eqref{eq:constraint2} for all $\psi_j$'s form the constraint set of the QP.

By solving the QP, we obtain the controllers for both robots. The trajectories of the robots are presented in Fig. \ref{fig:traj}. We make the following observations. First, both robots eventually reach their target regions, while avoiding the obstacle region. Furthermore, the second robot reaches region $B$ before the first robot reaches region $A$. This can be observed from Fig. \ref{fig:CBF}. CBF $h_{\psi_1}$ turns positive before $h_{\psi_2}$. Second, after the robots reaching their goal regions, they can still leave the goal region rather than remaining in the goal region. This is because our design of CBF adopts the guarding function $M_\psi(t)$ which increases over time and hence enhances the feasibility of the QP. The relaxation introduced by the guarding function enables the satisfaction of the connectivity constraint. As we could observe in Fig. \ref{fig:CBF}, CBFs $h_{\psi_1}$ and $h_{\psi_2}$ remain positive after the robots reaching their goal regions.

\section{Conclusion}\label{sec:conclusion}

In this paper, we studied the problem of control synthesis for CPSs under LTL constraints modeled by LTL without next operator. We focused on synthesizing the controller to satisfy the LTL constraint without explicitly computing the abstraction of the CPS. A CBF-based approach is used in this paper. We first constructed a sequence of LTL formulae corresponding to an accepting run on the DRA, and presented a design rule to design time-varying CBFs for the sequence of formulae. We introduced a function named guard function when designing CBFs, which enhances feasibility of CBF constraints. We showed that the positivity of the CBF implies the satisfaction of the LTL formula. Then we showed how to satisfy the set of formulae by guaranteeing the positivities of their CBFs. We formulated a QP to compute a controller that satisfies the LTL specification. A numerical case study is presented to illustrate the proposed approach. In future work, we aim to jointly consider the selection of accepting run and CBF-based control synthesis.


\bibliographystyle{IEEEtran}
\bibliography{IEEEabrv,MyBib}

\begin{thebibliography}{10}
\providecommand{\url}[1]{#1}
\csname url@rmstyle\endcsname
\providecommand{\newblock}{\relax}
\providecommand{\bibinfo}[2]{#2}
\providecommand\BIBentrySTDinterwordspacing{\spaceskip=0pt\relax}
\providecommand\BIBentryALTinterwordstretchfactor{4}
\providecommand\BIBentryALTinterwordspacing{\spaceskip=\fontdimen2\font plus
\BIBentryALTinterwordstretchfactor\fontdimen3\font minus
  \fontdimen4\font\relax}
\providecommand\BIBforeignlanguage[2]{{%
\expandafter\ifx\csname l@#1\endcsname\relax
\typeout{** WARNING: IEEEtran.bst: No hyphenation pattern has been}%
\typeout{** loaded for the language `#1'. Using the pattern for}%
\typeout{** the default language instead.}%
\else
\language=\csname l@#1\endcsname
\fi
#2}}

\bibitem{baier2008principles}
C.~Baier, J.-P. Katoen, and K.~G. Larsen, \emph{{Principles of Model
  Checking}}.\hskip 1em plus 0.5em minus 0.4em\relax MIT Press, 2008.

\bibitem{kress2009temporal}
H.~Kress-Gazit, G.~E. Fainekos, and G.~J. Pappas, ``Temporal-logic-based
  reactive mission and motion planning,'' \emph{Transactions on Robotics},
  vol.~25, no.~6, pp. 1370--1381, 2009.

\bibitem{coogan2015traffic}
S.~Coogan, E.~A. Gol, M.~Arcak, and C.~Belta, ``Traffic network control from
  temporal logic specifications,'' \emph{IEEE Transactions on Control of
  Network Systems}, vol.~3, no.~2, pp. 162--172, 2015.

\bibitem{alur2000discrete}
R.~Alur, T.~A. Henzinger, G.~Lafferriere, and G.~J. Pappas, ``Discrete
  abstractions of hybrid systems,'' \emph{Proceedings of the IEEE}, vol.~88,
  no.~7, pp. 971--984, 2000.

\bibitem{fainekos2009temporal}
G.~E. Fainekos, A.~Girard, H.~Kress-Gazit, and G.~J. Pappas, ``Temporal logic
  motion planning for dynamic robots,'' \emph{Automatica}, vol.~45, no.~2, pp.
  343--352, 2009.

\bibitem{ding2014optimal}
X.~Ding, S.~L. Smith, C.~Belta, and D.~Rus, ``Optimal control of markov
  decision processes with linear temporal logic constraints,''
  \emph{Transactions on Automatic Control}, vol.~59, no.~5, pp. 1244--1257,
  2014.

\bibitem{wongpiromsarn2009receding}
T.~Wongpiromsarn, U.~Topcu, and R.~M. Murray, ``Receding horizon temporal logic
  planning for dynamical systems,'' in \emph{the Proc. of Intl. Conf. on
  Decision and Control (CDC)}.\hskip 1em plus 0.5em minus 0.4em\relax IEEE,
  2009, pp. 5997--6004.

\bibitem{niu2019optimal}
L.~Niu and A.~Clark, ``Optimal secure control with linear temporal logic
  constraints,'' \emph{IEEE Transactions on Automatic Control}, 2019.

\bibitem{wolff2014optimization}
E.~M. Wolff, U.~Topcu, and R.~M. Murray, ``Optimization-based trajectory
  generation with linear temporal logic specifications,'' in \emph{2014 IEEE
  International Conference on Robotics and Automation (ICRA)}.\hskip 1em plus
  0.5em minus 0.4em\relax IEEE, 2014, pp. 5319--5325.

\bibitem{horowitz2014compositional}
M.~B. Horowitz, E.~M. Wolff, and R.~M. Murray, ``A compositional approach to
  stochastic optimal control with co-safe temporal logic specifications,'' in
  \emph{2014 IEEE/RSJ International Conference on Intelligent Robots and
  Systems}.\hskip 1em plus 0.5em minus 0.4em\relax IEEE, 2014, pp. 1466--1473.

\bibitem{papusha2016automata}
I.~Papusha, J.~Fu, U.~Topcu, and R.~M. Murray, ``Automata theory meets
  approximate dynamic programming: Optimal control with temporal logic
  constraints,'' in \emph{2016 IEEE 55th Conference on Decision and Control
  (CDC)}.\hskip 1em plus 0.5em minus 0.4em\relax IEEE, 2016, pp. 434--440.

\bibitem{kappen2005linear}
H.~J. Kappen, ``Linear theory for control of nonlinear stochastic systems,''
  \emph{Physical review letters}, vol.~95, no.~20, p. 200201, 2005.

\bibitem{wieland2007constructive}
P.~Wieland and F.~Allg{\"o}wer, ``Constructive safety using control barrier
  functions,'' \emph{IFAC Proceedings Volumes}, vol.~40, no.~12, pp. 462--467,
  2007.

\bibitem{wang2017safety}
L.~Wang, A.~D. Ames, and M.~Egerstedt, ``Safety barrier certificates for
  collisions-free multirobot systems,'' \emph{IEEE Transactions on Robotics},
  vol.~33, no.~3, pp. 661--674, 2017.

\bibitem{ames2016control}
A.~D. Ames, X.~Xu, J.~W. Grizzle, and P.~Tabuada, ``Control barrier function
  based quadratic programs for safety critical systems,'' \emph{IEEE
  Transactions on Automatic Control}, vol.~62, no.~8, pp. 3861--3876, 2016.

\bibitem{srinivasan2019control}
M.~Srinivasan and S.~Coogan, ``Control of mobile robots using barrier functions
  under temporal logic specifications,'' \emph{arXiv preprint
  arXiv:1908.04903}, 2019.

\bibitem{srinivasan2018control}
M.~Srinivasan, S.~Coogan, and M.~Egerstedt, ``Control of multi-agent systems
  with finite time control barrier certificates and temporal logic,'' in
  \emph{2018 IEEE Conference on Decision and Control (CDC)}.\hskip 1em plus
  0.5em minus 0.4em\relax IEEE, 2018, pp. 1991--1996.

\bibitem{lindemann2018control}
L.~Lindemann and D.~V. Dimarogonas, ``Control barrier functions for signal
  temporal logic tasks,'' \emph{IEEE control systems letters}, vol.~3, no.~1,
  pp. 96--101, 2018.

\bibitem{yang2019continuous}
G.~Yang, R.~Tron, and C.~Belta, ``Continuous-time signal temporal logic
  planning with control barrier function,'' \emph{arXiv preprint
  arXiv:1903.03860}, 2019.

\bibitem{boyd2004convex}
S.~Boyd and L.~Vandenberghe, \emph{Convex optimization}.\hskip 1em plus 0.5em
  minus 0.4em\relax Cambridge university press, 2004.

\bibitem{wongpiromsarn2015automata}
T.~Wongpiromsarn, U.~Topcu, and A.~Lamperski, ``Automata theory meets barrier
  certificates: Temporal logic verification of nonlinear systems,'' \emph{IEEE
  Transactions on Automatic Control}, vol.~61, no.~11, pp. 3344--3355, 2015.

\bibitem{xu2018constrained}
X.~Xu, ``Constrained control of input--output linearizable systems using
  control sharing barrier functions,'' \emph{Automatica}, vol.~87, pp.
  195--201, 2018.

\bibitem{li2018formally}
A.~Li, L.~Wang, P.~Pierpaoli, and M.~Egerstedt, ``Formally correct composition
  of coordinated behaviors using control barrier certificates,'' in \emph{2018
  IEEE/RSJ International Conference on Intelligent Robots and Systems
  (IROS)}.\hskip 1em plus 0.5em minus 0.4em\relax IEEE, 2018, pp. 3723--3729.

\bibitem{haddad2008finite}
W.~M. Haddad, S.~G. Nersesov, and L.~Du, ``Finite-time stability for
  time-varying nonlinear dynamical systems,'' in \emph{2008 American control
  conference}.\hskip 1em plus 0.5em minus 0.4em\relax IEEE, 2008, pp.
  4135--4139.

\bibitem{babiak2013effective}
T.~Babiak, F.~Blahoudek, M.~K{\v{r}}et{\'\i}nsk{\`y}, and J.~Strej{\v{c}}ek,
  ``Effective translation of ltl to deterministic rabin automata: Beyond the
  ({F}, {G})-fragment,'' in \emph{Automated Technology for Verification and
  Analysis}.\hskip 1em plus 0.5em minus 0.4em\relax Springer, 2013, pp. 24--39.

\bibitem{ames2014control}
A.~D. Ames, J.~W. Grizzle, and P.~Tabuada, ``Control barrier function based
  quadratic programs with application to adaptive cruise control,'' in
  \emph{53rd IEEE Conference on Decision and Control}.\hskip 1em plus 0.5em
  minus 0.4em\relax IEEE, 2014, pp. 6271--6278.

\end{thebibliography}

\end{document}